\documentclass[conference,10pt]{IEEEtran}

\usepackage{times}
\usepackage[final]{graphicx}
\usepackage{amsmath,amsfonts}
\usepackage{amssymb,amsbsy}
\usepackage{times}

\usepackage{epsfig}
\usepackage{cite}

\newcommand\ignore[1]{}

\newcommand{\hsp}{\hspace{0.1in} }
\newcommand{\hspp}{\hspace{0.05in} }
\newcommand{\hsppp}{\hspace{0.02in} }

\newcommand{\snr}{ {\sf SNR}  }

\newcommand{\iid} {  {\sf{iid}} }

\newcommand{\bggi} { {\bf h}_{ {\sf iid}, \hsppp i} }

\newcommand{\bha} { {  \mathbf{h}  } }

\newcommand{\bI} {  {\mathbf{I}} }

\newtheorem{prp}{Proposition}
\newtheorem{thm}{Theorem}

\title{\Large {\bfseries{Linear Beamforming for the Spatially Correlated
MISO Broadcast Channel}}
{\vspace{-0.2in}}
\large \author{Vasanthan Raghavan$^{\dagger}$, Venugopal Veeravalli$^{\star}$,
Stephen Hanly$^{\dagger}$ \\
$^{\dagger}$The University of Melbourne, $^{\star}$University of Illinois \\
\normalsize \{vasanthan\_\_raghavan@ieee.org, vvv@illinois.edu,
hanly@unimelb.edu.au\}
}}

\begin{document}
\maketitle
\thispagestyle{empty}
\pagestyle{empty}

\begin{abstract}
\noindent
A spatially correlated broadcast setting with $M$ antennas at the
base station and $M$ users (each with a single antenna) is considered.
We assume that the users have perfect channel information about
their links and the base station has only statistical information
about each user's link. The base station employs a linear beamforming strategy
with one spatial eigen-mode allocated to each user. The goal of this
work is to understand the structure of the beamforming vectors that
maximize the ergodic sum-rate achieved by treating interference as
noise. In the $M = 2$ case, we first fix the beamforming vectors and
compute the ergodic sum-rate in closed-form as a function of the channel
statistics. We then show that the optimal beamforming vectors are the
dominant generalized eigenvectors of the covariance matrices of the
two links. It is difficult to obtain intuition on the structure of the
optimal beamforming vectors for $M > 2$ due to the complicated
nature of the sum-rate expression. Nevertheless, in the case of asymptotic
$M$, we show that the optimal beamforming vectors have to satisfy a
set of fixed-point equations. 
\end{abstract}

\section{Introduction}
The focus of this paper is on a MISO broadcast (downlink) setting
where the base station (BS) has $M$ antennas with $M$ users in the
cell, each having a single antenna. Under the assumption of perfect
channel state information (CSI) at both the ends, significant progress
has been made over the last few years on understanding optimal
signaling that achieves the
sum-capacity~\cite{caire_shamai,pramodv,jindal1,jindal2,wei_yu} as
well as the capacity region~\cite{weingarten} of the multi-antenna
broadcast channel. Though the capacity-achieving {\em dirty paper coding}
scheme is well-understood, the complexity associated with it makes it an
impractical choice. Thus, recent focus has been on a family of linear
precoding schemes~\cite{boche1,swindlehurst,wiesel} which are within a
fixed power-offset of the dirty paper coding scheme. In particular, a
linear beamforming scheme that allocates one eigen-mode to each user is
of considerable interest in standardization efforts.

More importantly, while reasonably accurate CSI can be obtained at the
users via pilot-based schemes, CSI at the BS requires either channel
reciprocity or reverse link feedback, both of which put an overwhelming
burden on the operating cost. Thus, there has been a significant
interest on understanding the information-theoretic limits of broadcast
channels under practical assumptions on CSI. In the extreme case of no
CSI at the BS, the multiplexing gain possible in the perfect CSI case ($M$)
is lost completely as it reduces to one.

The no CSI assumption is pessimistic and in practice, the channel evolves
fairly slowly on a statistical
scale and it is possible to learn the statistics of the individual links
at the BS with minimal cost. In the MISO broadcast setting with a Rayleigh
fading model for each user (zero mean complex Gaussian fading process),
the complete channel statistics are specified by the covariance matrix
of the vector channel of the user. In this context, it must be noted that
initial works assume that all the users experience fading that is
{\em independent and identically distributed} (i.i.d.) across the antennas.
That is, the covariance matrix of each user is the identity matrix. This
assumption cannot be justified in practice unless the antennas at the BS
are spaced wide apart and the scattering environment connecting the BS
with the users is rich. While the correlated case
has been studied in the literature~\cite{tareq,trivellato}, the general
version of the problem studied here has not received much attention.

The focus of this work is on understanding the impact of the users'
spatial statistics (their covariance matrices) on the sum-rate
performance of the linear beamforming scheme. We first study the
simplest non-trivial case of $M = 2$ and compute the sum-rate achievable
with a linear beamforming scheme under the practical assumption that
interference is treated as noise. For this, we exploit knowledge of the
structure of density function of the weighted norm of isotropically
distributed beamforming vectors~\cite{vasanth_isit_rhv}. Our sum-rate
characterization is explicit and in terms of the covariance matrices of
the two users and the beamforming vectors.

While identifying the structure of the sum-rate optimizing beamforming
vectors is a difficult problem, in general, we obtain intuition in the
low- and the high-$\snr$ extremes. In the low-$\snr$ extreme, it
is not surprising that a strategy where the BS beamforms along the
dominant eigen-mode of each user's channel is sum-rate optimal. In the high-$\snr$
extreme, a strategy where the BS beamforms to a given user along the
dominant generalized eigenvector\footnote{ \label{fn_general}
A generalized eigenvector ${\bf x}$ (with the corresponding generalized
eigenvalue $\sigma$) of a pair of matrices $\left( {\bf A}, \hsppp
{\bf B} \right)$ satisfies the relationship ${\bf A} \hsppp {\bf x} =
\sigma \hsppp {\bf B} \hsppp {\bf x}.$ In the special case where ${\bf B}$
is invertible, a generalized eigenvector of the pair $\left( {\bf A},
\hsppp {\bf B} \right)$ is an eigenvector of ${\bf B}^{-1} {\bf A}$. If
${\bf A}$ and ${\bf B}$ are also positive definite, then all the
generalized eigenvalues are also positive.} of that user's and
the other user's covariance matrices is sum-rate optimal. Intuitively
speaking, given that the BS has only
statistical information of the two links, it generates an ``effective''
covariance matrix for a particular user by statistically pre-nulling
the interference from the forward channel of the user. The sum-rate
optimal beamforming vectors are the dominant eigen-modes of these
effective covariance matrices. Solutions in terms of the generalized
eigenvectors are obtained in the perfect CSI case~\cite{wiesel,coord_bf},
but to the best of our knowledge, this solution in the statistical case
is a first.
While the generalization of this result to the $M > 2$ case is
cumbersome, simple approximations for the ergodic sum-rate in terms
of the channel statistics are provided in the asymptotics of $M$.
Based on these approximations, we show that the optimal beamforming
vectors are solutions to a set of fixed-point equations.

\noindent {\bf \em Note:} Due to space constraints, the proofs of the
main statements in this paper are not provided and the logic of the main
arguments are sketched out in brief.

\section{System Setup}
\label{sec2}
We consider a broadcast setting with $M$ antennas at the
base station (BS) and $M$ users, each with a single antenna. We denote
the $M \times 1$ channel between the BS and user $i$ as $\bha_i,
\hsppp i = 1, \cdots, M$. While different multi-user communication
strategies can be considered, as motivated in the Introduction, the
focus here is on a linear beamforming scheme where the BS beamforms the
information-bearing signal $s_i$ meant for user $i$ with the $M \times 1$
unit-normed vector ${\bf w}_i$. We assume that $s_i$ is unit energy and
the BS divides its power budget of $\rho$ equally across all the users.
The received symbol $y_i$ at user $i$ is written as
\begin{eqnarray}
y_ i & = & \sqrt{ \frac{ \rho}{M} } \cdot
\bha_i^H \left( \sum_{i=1}^M {\bf w}_i s_i \right) + n_i, \hspp
i = 1, \cdots , M  \nonumber
\end{eqnarray}
where $n_i$ denotes the ${\mathcal{CN}}(0,1)$ complex Gaussian noise
added at the receiver.

We assume a Rayleigh fading (zero mean complex Gaussian) model for the
channel and hence, the complete spatial statistics are described by the
second-order moments of $\{ \bha_i \}$. With $M$ antennas at the BS and a
single antenna at each user, the channel $\bha_i$ of user $i$ can be
generically written as
\begin{eqnarray}
\bha_i & = & {\bf \Sigma}_{i}^{1/2} \hsppp \bggi 
\label{bhai}
\end{eqnarray}
where $\bggi$ is an $M \times 1$ vector with i.i.d.\ ${\mathcal{ CN}}(0,1)$
entries and ${\bf{\Sigma}}_{i}$ is the covariance matrix corresponding
to the user $i$.
In particular, with ${\bf \Sigma}_i = \bI_{M}$ for all users,~(\ref{bhai})
reduces to the i.i.d.\ downlink model well-studied in the literature.

The metric of interest in this work is the throughput from the BS to the
users. Under the assumption of Gaussian inputs $\{s_i \}$,
the instantaneous information-theoretic\footnote{All rate quantities
will be assumed to be in nats$/$s$/$Hz in this work.} rate, $R_i$,
achievable by user $i$ with the linear beamforming scheme and using a
mismatched decoder is given by
\begin{align}
& R_i 
= \log \left( 1 + \frac{ \frac{\rho}{M} \cdot
| {\bf h}_i^H {\bf w}_i |^2  }
{ 1 + \frac{ \rho}{M} \cdot \sum_{j \neq i} | {\bf h}_i^H {\bf w}_j  |^2 }
\right) \nonumber \\
& 
=  \underbrace{ \log \left( 1 + \frac{\rho}{M}
\sum_{j=1}^M | {\bf h}_i^H {\bf w}_j |^2 \right)}_{I_{i, \hsppp 1} }
- \underbrace{ \log \left( 1 + \frac{\rho}{M}
 \sum_{j \neq i} | {\bf h}_i^H {\bf w}_j |^2 \right)
}_{ I_{i,\hsppp 2} }. \nonumber
\end{align}
In particular, the ergodic sum-rate achievable with the linear
beamforming scheme is given by
\begin{eqnarray}
{\cal R} \triangleq \sum_{i=1}^M E \left[ R_i \right].
\nonumber
\end{eqnarray}
With the spatial correlation model assumed in~(\ref{bhai}), we can write
$I_{i,\hsppp 1}$ as
\begin{align}
& I_{i, \hsppp 1} = \log \left( 1 + \frac{\rho}{M} \cdot
\bggi^H \hsppp {\bf \Sigma}_{i}^{1/2} \left( \sum_{j=1}^M {\bf w}_j
{\bf w}_j^H \right) {\bf \Sigma}_{i}^{1/2} \hsppp \bggi \right)
\nonumber \\
& {\hspace{0.2in}}
= \log \left( 1 +  \frac{\rho}{M} \cdot \bggi^H \hsppp {\bf V}_i
\hsppp {\bf \Lambda}_i \hsppp {\bf V}_i^H \bggi \right),
\nonumber
\end{align}
where we have used the following eigen-decomposition in the second equation:
\begin{align}
& {\hspace{0.65in}}
{\bf V}_i \hsppp {\bf \Lambda}_i \hsppp {\bf V}_i^H =
{\bf \Sigma}_{i}^{1/2} \left( \sum_{j=1}^M {\bf w}_j {\bf w}_j^H \right)
{\bf \Sigma}_{i}^{1/2}
\label{eqna7}
\nonumber
\\ & {\hspace{0.1in}}
{\bf \Lambda}_i = {\sf diag}\Big( [ {\bf \Lambda}_{i, \hsppp 1},
\hspp \cdots, \hspp {\bf \Lambda}_{i, \hsppp M}] \Big), \hspp
{\bf \Lambda}_{i,\hsppp 1} \geq \cdots \geq {\bf \Lambda}_{i, \hsppp M}
\geq 0.
\nonumber
\end{align}
Similarly, we can write $I_{i, \hsppp 2}$ as
\begin{align}
& {\hspace{0.55in}}
I_{i, \hsppp 2} = \log \left( 1 +  \frac{\rho}{M} \cdot \bggi^H
\hsppp \widetilde{ {\bf V}} _i \hsppp \widetilde{ {\bf \Lambda}}_ i
\hsppp \widetilde{ {\bf V}}_i^H \bggi \right) \nonumber \\
& {\hspace{0.65in}}
\widetilde{ {\bf V}} _i \hsppp \widetilde{ {\bf \Lambda}}_ i \hsppp
\widetilde{ {\bf V}}_i^H = {\bf \Sigma}_{i}^{1/2} \left( \sum_{j \neq i}
{\bf w}_j {\bf w}_j^H \right) {\bf \Sigma}_{i}^{1/2}
\nonumber \\
& {\hspace{0.1in}}
\widetilde{{ \bf \Lambda}}_i = {\sf diag}\left([
\widetilde{ {\bf \Lambda}} _{i, \hsppp 1},  \hspp \cdots, \hspp
\widetilde{ {\bf \Lambda}}_{i, \hsppp M}] \right),
\hspp \widetilde{{\bf \Lambda}}_{i,\hsppp 1} \geq \cdots
\geq \widetilde{ {\bf \Lambda}} _{i, \hsppp M} \geq 0.
\nonumber
\end{align}

Towards the goal of computing the ergodic rates, we expand
${\bf h}_{\iid, \hsppp i }$ into its magnitude and directional components
as ${\bf h}_{\iid, \hsppp i } = \| {\bf h}_{\iid, \hsppp i} \| \cdot
\widetilde{ {\bf h} }_{\iid, \hsppp i}$. Note that
$\| {\bf h}_{\iid, \hsppp i} \|^2$ can be written as
\begin{eqnarray}
\| {\bf h}_{\iid, \hsppp i} \|^2 = \frac{1}{2} \sum_{j=1}^{2M} z_j^2
\nonumber
\end{eqnarray}
where $z_j^2$ is a standard (real) chi-squared random variable and
$\widetilde{ {\bf h} }_{\iid, \hsppp i}$ is a unit-normed vector that
is isotropically distributed on the surface of $M$-dimensional complex
sphere. Thus, we can rewrite $I_{i, \hsppp 1}$ and $I_{i, \hsppp 2}$ as
\begin{eqnarray}
I_{i, \hsppp 1} & = &
\log \left( 1 +  \frac{\rho}{M} \cdot \| {\bf h}_{ \iid, \hsppp i } \|^2
\cdot \widetilde{ {\bf h} }_{\iid, \hsppp i}^H \hsppp {\bf V}_i
\hsppp {\bf \Lambda}_i \hsppp {\bf V}_i^H \widetilde{ {\bf h} }_{\iid, \hsppp i}
\right) \nonumber \\
I_{i, \hsppp 2} & = &
\log \left( 1 +  \frac{\rho}{M} \cdot \| {\bf h}_{ \iid, \hsppp i } \|^2
\cdot \widetilde{ {\bf h} }_{\iid, \hsppp i}^H
\hsppp \widetilde{ {\bf V}} _i
\hsppp \widetilde{ {\bf \Lambda}}_ i \hsppp \widetilde{ {\bf V}}_i^H
\widetilde{ {\bf h} }_{\iid, \hsppp i} \right).
\nonumber
\end{eqnarray}
Further, since the magnitude and directional information of an i.i.d.\
(isotropically distributed) random vector are independent,
$E \left[I_{i,\hsppp 1} \right]$ and $E \left[I_{i,\hsppp 2} \right]$
can be written as
\begin{align}
& E\left[I_{i, \hsppp 1} \right] = E
\left[  \log \left( 1 + \frac{\rho}{M} \cdot
\| {\bf h}_{ \iid, \hsppp i } \|^2 \cdot
\widetilde{ {\bf h} }_{\iid, \hsppp i}^H \hsppp
 {\bf \Lambda}_ i  \hsppp
\widetilde{ {\bf h} }_{\iid, \hsppp i} \right)
\right] 
\nonumber \\
& E\left[I_{i, \hsppp 2} \right] = E
\left[  \log \left( 1 + \frac{\rho}{M} \cdot
\| {\bf h}_{ \iid, \hsppp i } \|^2 \cdot
\widetilde{ {\bf h} }_{\iid, \hsppp i}^H \hsppp
\widetilde{ {\bf \Lambda}}_ i  \hsppp \widetilde{ {\bf h} }_{\iid, \hsppp i} \right)
\right] 
\nonumber
\end{align}
where we have also used the fact that a fixed\footnote{Note that the
unitary transformation is independent of the channel realization when
the beamforming vectors are chosen based on long-term statistics of the
channel.} unitary transformation of an isotropically distributed vector
on the surface of the complex sphere does not alter its distribution.

\section{Ergodic Sum-Rate: Two User Case}
\label{sec3}
The focus of this section is on computing the ergodic
information-theoretic rates in closed-form in the special case of
two users ($M = 2$). This closed-form expression will be a function
of the covariance matrices of the two users, ${\bf \Sigma}_1$
and ${\bf \Sigma}_2$, and the choice of beamforming vectors, ${\bf w}_1$
and ${\bf w}_2$. Once a closed-form expression is obtained, our goal
lies in characterizing the structure of the optimal beamforming vectors
as a function of the channel statistics and ${\sf SNR}$.

For simplicity, we assume that
\begin{eqnarray}
{\bf \Sigma}_1 = {\bf U} \hsppp {\sf diag}([\lambda_1 \hspp \lambda_2])
\hsppp {\bf U}^H, \hspp
{\bf \Sigma}_2 = \widetilde{\bf U} \hsppp {\sf diag}([\mu_1 \hspp \mu_2])
\hsppp \widetilde{\bf U}^H
\label{pdef}
\end{eqnarray}
where ${\bf U} = \left[ {\bf u}_1 ({\bf \Sigma}_1), \hspp
{\bf u}_2({\bf \Sigma}_1)  \right]$, ${\widetilde{ {\bf U}}} =
\left[ {\bf u}_1({\bf \Sigma}_2), \hspp {\bf u}_2({\bf \Sigma}_2)
\right]$, $\lambda_1 \geq \lambda_2 > 0$ and $\mu_1 \geq \mu_2 > 0$
(that is, both ${\bf \Sigma}_1$ and ${\bf \Sigma}_2$ are positive
definite). Define the condition numbers $\kappa_1$ and $\kappa_2$ as
\begin{eqnarray}
\kappa_1 \triangleq \frac{\lambda_1}{\lambda_2} \hspp \hspp {\rm and }
\hspp \hspp \kappa_2 \triangleq \frac{\mu_1}{\mu_2}.
\label{qdef}
\end{eqnarray}
\begin{prp}
\label{prop_basic_rate}
The ergodic information-theoretic rate achievable at user $i$ (where
$i = 1,2$) with linear beamforming in the two user case is given by
\begin{eqnarray}
\begin{split}
& E \left[R_{i} \right] =
E \left[  I_{i, \hsppp 1} \right] - E \left[ I_{i, \hsppp 2}  \right]
\\ & {\hspace{0.2in}}
=  \frac{  {\bf \Lambda}_{i, \hsppp 1} \cdot
e^{ \frac{2}{\rho {\bf \Lambda}_{i, \hsppp 1}}  }
E_1 \left( \frac{2}{\rho {\bf \Lambda}_{i, \hsppp 1}}  \right)
- {\bf \Lambda}_{i, \hsppp 2} \cdot
e^{ \frac{2}{\rho {\bf \Lambda}_{i, \hsppp 2}} }
E_1 \left( \frac{2}{\rho {\bf \Lambda}_{i, \hsppp 2}}  \right) }
{  {\bf \Lambda}_{i, \hsppp 1} - {\bf \Lambda}_{i, \hsppp 2} }
\\ & {\hspace{1.5in}} -
\exp \left(  2/\rho \widetilde{ {\bf \Lambda}} _{i, \hsppp 1}  \right)
E_1 \left( 2/\rho \widetilde{ {\bf \Lambda}} _{i, \hsppp 1}  \right)
\nonumber
\end{split}
\end{eqnarray}
where $E_1(x) = \int_{x}^{\infty} \frac{e^{-t}}{t} dt$ is the exponential
integral. The corresponding eigenvalues can be written in terms of
${\bf \Sigma}_i$ and the beamforming vectors as follows:
\begin{eqnarray}
{\bf \Lambda}_{i, \hsppp 1} & = &
\frac{ A_i + B_i + \sqrt{ (A_i - B_i)^2 + 4C_i^2 }  }{2}
\nonumber \\
{\bf \Lambda}_{i, \hsppp 2} & = &
\frac{ A_i + B_i - \sqrt{ (A_i - B_i)^2 + 4C_i^2 }  }{2}
\nonumber \\
\widetilde{ {\bf \Lambda}} _{i, \hsppp 1} & = & B_i \nonumber \\
\widetilde{ {\bf \Lambda}} _{i, \hsppp 2} & = & 0 \nonumber
\end{eqnarray}
where $A_i = {\bf w}_i^H {\bf \Sigma}_i {\bf w}_i$, $B_i = {\bf w}_j^H
{\bf \Sigma}_i {\bf w}_j$ and $C_i = |{\bf w}_i^H {\bf \Sigma}_i
{\bf w}_j|$ with $j \neq i$ and $\{ i,j \} = 1,2$.
\end{prp}
\begin{proof}
Note that a closed-form computation of $E \left[ I_{i, \hsppp 1} \right]$
requires the density function of weighted norms of isotropically distributed
unit-normed vectors since
\begin{eqnarray}
E \left[I_{i, \hsppp 1} \right] =
E_{ {\bf X} } \left[ \int_{  y = {\bf \Lambda}_{i,\hsppp 2} }
^{{\bf \Lambda}_{i, \hsppp 1} }
\log \left( 1 + X y  \right) {\sf P}_i(y) dy \right] \nonumber
\end{eqnarray}
where ${\bf X}$ denotes the random variable ${\bf X} = \frac{\rho}{2} \cdot
\| {\bf h}_{ \iid, \hsppp i } \|^2$ and $X$ corresponds to a realization of
${\bf X}$. Let ${\bf Y}$ denote the random variable
\begin{eqnarray}
{\bf Y} \triangleq
\widetilde{ {\bf h} }_{\iid, \hsppp i}^H \hsppp {\bf \Lambda}_i \hsppp
\widetilde{ {\bf h} }_{\iid, \hsppp i}
= \sum_{j=1}^2 {\bf \Lambda}_{i,\hsppp j} \hspp \Big| \widetilde{
{\bf h} }_{\iid, \hsppp i}(j) \Big|^2.
\nonumber
\end{eqnarray}
The Ritz-Rayleigh relationship implies that
${\bf \Lambda}_{i,\hsppp 2}  \leq {\bf Y} \leq {\bf \Lambda}_{i,\hsppp 1}$
and the density function of ${\bf Y}$ evaluated at $y$ is denoted as
${\sf P}_i(y)$. In~\cite{vasanth_isit_rhv}, ${\sf P}_i(y)$
in the $M = 2$ case is shown to be uniform, that is,
\begin{eqnarray}
{\sf P}_i(y) = \frac{1}{ {\bf \Lambda}_{i, \hsppp 1} -
{\bf \Lambda}_{i, \hsppp 2} }, \hspp {\bf \Lambda}_{i, \hsppp 2} \leq
y \leq {\bf \Lambda}_{i, \hsppp 1}. \nonumber
\end{eqnarray}
The statement of the proposition follows from a routine computation
via the integral tables~\cite{gradshteyn}.
\end{proof}
Note that understanding the structure of the optimal choice of beamforming
vectors, $\left( {\bf w}_{1, \hsppp {\sf opt}}, {\bf w}_{2, \hsppp {\sf opt} }
\right)$, that maximize the ergodic sum-rate as a function of
${\bf \Sigma}_1$, ${\bf \Sigma}_2$ and $\rho$ is a hard problem, in general.
Therefore, we consider the low- and the high-$\snr$ extremes to obtain
insights.

\noindent {\bf \em Low-$\snr$ Extreme:} We need the following
characterization of the exponential integral:
\begin{eqnarray}
\frac{1}{x+2} \leq \frac{1}{2} \log \left( 1 + \frac{2}{x} \right)
\leq E_1(x)e^{x} \leq \log \left(1 + \frac{1}{x}  \right) \leq
 \frac{1}{x}
\nonumber
\end{eqnarray}
where the extremal inequalities are established by using the fact that
$\frac{x}{x+1} \leq \log(1 + x) \leq x$. Note that the upper and the lower
bounds get tight as $x \rightarrow \infty$ (or $\rho \rightarrow 0$ in this
context). Using the above bound, we have
\begin{eqnarray}
E \left[R_i \right] \stackrel{\rho \rightarrow 0}{\rightarrow} \frac{\rho}{2}
\left( {\bf \Lambda}_{i, \hsppp 1} + {\bf \Lambda}_{i, \hsppp 2} - B_i
\right) = \frac{\rho}{2} \cdot A_i = \frac{\rho}{2} \cdot
{\bf w}_i^H {\bf \Sigma}_i {\bf w}_i. \nonumber
\end{eqnarray}
In the low-$\snr$ regime, the system is noise-limited and hence, the
linear scaling of $E \left[R_i  \right]$ with $\snr$.
It is also straightforward to note that maximizing $E \left[ R_i \right]$
is contingent on optimizing over ${\bf w}_i$ alone. Thus, 
the sum-rate is maximized by
\begin{eqnarray}
{\bf w}_{1, \hsppp {\sf opt}} = {\bf u}_1( {\bf \Sigma}_1 ) \hsp
{\rm and} \hsp
{\bf w}_{2, \hsppp {\sf opt}} = {\bf u}_1( {\bf \Sigma}_2 )
\nonumber
\end{eqnarray}
where 
${\bf u}_1( {\bf \Sigma}_i)$ denotes the dominant eigenvector (an
eigenvector corresponding to the dominant eigenvalue) of ${\bf \Sigma}_i$.
In other words, in the low-$\snr$ extreme, each user signals along the
optimal statistical eigen-mode of its channel (and
ignoring the other user's channel completely). This conclusion should
not be entirely surprising. The resulting ergodic sum-rate is given as
\begin{eqnarray}
{\cal R} 
\stackrel{\rho \rightarrow 0}{ \rightarrow}
\frac{\rho}{2} \cdot \Big[  \lambda_{\max}( {\bf \Sigma}_1) +
 \lambda_{\max} ( {\bf \Sigma}_2) \Big]. \nonumber
\end{eqnarray}

\noindent {\bf \em High-$\snr$ Extreme:} The following expansion
of the exponential
integral is useful in characterizing ${\cal R}$ 
as $\rho \rightarrow \infty$:
\begin{eqnarray}
E_1(x) & = & \log \left( \frac{1}{x} \right)
+ \sum_{k = 1}^{\infty} \frac{ (-1)^{k+1} x^k }{k \cdot k!} -\gamma
\nonumber
\\
& \stackrel{x \rightarrow 0}{ \rightarrow } & \log \left( \frac{1}{x} \right)
+ x - \gamma
\nonumber
\end{eqnarray}
where $\gamma \approx 0.577$ is the Euler-Mascheroni constant. Using the
above approximation, we have
\begin{eqnarray}
E \left[ R_i \right] \stackrel{ \rho \rightarrow \infty }
{ \rightarrow} \frac{  {\bf \Lambda}_{i,\hsppp 1}
\log \left(  {\bf \Lambda}_{i, \hsppp 1}  \right)
- {\bf \Lambda}_{i,\hsppp 2}
\log \left( {\bf \Lambda}_{i, \hsppp 2}  \right)  }
{ {\bf \Lambda}_{i, \hsppp 1} - {\bf \Lambda}_{i, \hsppp 2}  } -
\log \left( B_i \right).
\nonumber 
\end{eqnarray}
The dominating impact of interference (due to the fixed nature of the
linear beamforming scheme where the beamforming vectors are not adapted
to the channel realizations) and the consequent boundedness of
$E \left[R_i \right]$ as $\snr$ increases should not be
surprising. After some elementary manipulation, we can write
$E \left[ R_i \right]$ as
\begin{eqnarray}
2 E \left[ R_i \right] & \stackrel{\rho \rightarrow \infty}{\rightarrow} &
\log \left( \frac{ A_i B_i - C_i^2} {B_i^2}  \right) +
\frac{ A_i + B_i  }{
\sqrt{ \left( A_i - B_i \right)^2 + 4C_i^2}  } \cdot
\nonumber \\
& & \log \left(  \frac{ A_i + B_i + \sqrt{ \left( A_i - B_i \right)^2 + 4C_i^2}  }
{A_i + B_i - \sqrt{ \left( A_i - B_i \right)^2 + 4C_i^2} } \right)
\nonumber \end{eqnarray}
We now rewrite the high-$\snr$ ergodic rates in a form that eases
further study.
\begin{prp}
\label{prop2}
Define $d_{ {\bf \Sigma}_i } \left(
{\bf w}_1, {\bf w}_2 \right)$ between two unit-normed vectors
${\bf w}_1$ and ${\bf w}_2$ from the Grassmann manifold ${\cal G}(2,1)$
as
\begin{eqnarray}
d_{ {\bf \Sigma}_i } \left( {\bf w}_1, {\bf w}_2  \right) \triangleq
\sqrt{ \frac{ 4 \left(A_i B_i - C_i^2  \right) }
{ \left(A_i + B_i  \right)^2 } },
\nonumber
\end{eqnarray}
where $A_i, B_i$ and $C_i$ are as in the statement of
Prop.~\ref{prop_basic_rate}.
\begin{itemize}
\item
(a) Then, $d_{ {\bf \Sigma}_i } \left( {\bf w}_1, {\bf w}_2  \right)$
is a generalized ``distance'' semi-metric\footnote{A semi-metric satisfies
all the properties necessary for a distance metric, except the triangle
inequality.} between ${\bf w}_1$ and ${\bf w}_2$
satisfying $0 \leq d_{ {\bf \Sigma}_i } \left( \cdot, \cdot  \right)
\leq 1$.
\item
(b) We can recast the ergodic rate in terms of $d_{ {\bf \Sigma}_i }
\left( {\bf w}_1, {\bf w}_2  \right)$ as
\begin{eqnarray}
 E \left[ R_i \right] + \log(2)
\stackrel{\rho \rightarrow \infty}{ \rightarrow }
\frac{g\left(d_{ {\bf \Sigma}_i } ( {\bf w}_1, {\bf w}_2  )
\right)}{2} +  \log \left(1 + \frac{A_i}{B_i} \right)
\label{ri_highsnr}
\nonumber
\end{eqnarray}
where
\begin{eqnarray}
\label{gstruct}
\nonumber
g(z) & = &  f(z) + 2\log(z),  \\
f(z) & = & \frac{1}{\sqrt{1 - z^2}} \log \left(
\frac{1 + \sqrt{1 - z^2} }{1 - \sqrt{1 - z^2}} \right).
\label{fstruct}
\nonumber
\end{eqnarray}
\item
(c) While $f (\bullet)$ is monotonically decreasing as a function of its
argument, $g(\bullet)$ is increasing with
\begin{eqnarray}
2\log(2) = \lim_{z \rightarrow 0} g(z)
\leq & g(z) &  \leq \lim_{z \rightarrow 1} g(z) = 2 \nonumber \\
\infty = \lim_{z \rightarrow 0} f(z) \geq & f(z) & \geq
\lim_{z \rightarrow 1} f(z) = 2.
\nonumber \end{eqnarray}
\end{itemize}
\endproof
\end{prp}

We are now prepared to illustrate the structure of the optimal beamforming
vectors.
\begin{thm}
\label{prop1}
The optimal choice of the pair $\left( {\bf w}_{1, \hsppp {\sf opt}},
{\bf w}_{2, \hsppp {\sf opt}} \right)$ that maximizes
$E \left[  R_i \right]$ in the high-$\snr$ regime is
\begin{eqnarray}
{\bf w}_{i, \hsppp {\sf opt}} = e^{j \nu_1} \hsppp {\bf u}_1
\left( {\bf \Sigma}_i  \right)
\hsp {\rm and} \hsp {\bf w}_{j, \hsppp {\sf opt}} = e^{j \nu_2}
\hsppp {\bf u}_2 \left( {\bf \Sigma}_i  \right), \hspp j \neq i
\label{opt_prop3}
\nonumber
\end{eqnarray}
for some choice of $\nu_i \in [0, 2\pi), \hsppp i = 1,2$.
\end{thm}
\begin{proof}
Let $\chi \left( {\bf \Sigma}_i \right) =
\frac{ \lambda_{\max}( {\bf \Sigma}_i) } { \lambda_{\min}(
{\bf \Sigma}_i )}$ denote the condition number of ${\bf \Sigma}_i$. We
first note that the optimization problem over the choice of a pair
$\left( {\bf w}_1, {\bf w}_2 \right)$ that results in a corresponding
choice of $\left( A_i, B_i, C_i \right)$ can be recast in the form of a
two parameter optimization problem over $(M_i, N_i)$ with
$M_i = \frac{A_i}{B_i}$, $N_i = \frac{C_i}{B_i}$ under the constraint
that $0 \leq N_i^2 \leq M_i \leq \chi \left( {\bf \Sigma}_i \right)$.
This results in the following high-$\snr$ expression:
\begin{eqnarray}
2 E \left[ R_i \right] + 2 \log(2)
= g\left( \frac{2 \sqrt{ M_i - N_i^2} } { M_i+1} \right) +
2 \log \left(1 + M_i \right).
\nonumber
\end{eqnarray}
It is straightforward to show that the choice in the theorem
maximizes the above equation.
\end{proof}
With this choice of beamforming vectors, $d_{ {\bf \Sigma}_i }
(\cdot, \cdot)$ and $E \left[ R_i \right]$ can be written as
\begin{eqnarray}
d_{ {\bf \Sigma}_i } \left( {\bf w}_{i, \hsppp {\sf opt}},
{\bf w}_{j, \hsppp {\sf opt}} \right)
& = & \frac{2 \sqrt{ \kappa_i}}{ \kappa_i+1 } \nonumber \\
E \left[ R_i \right]
& \stackrel{\rho \rightarrow \infty}{ \rightarrow } &
\frac{\kappa_i \log(\kappa_i)}{\kappa_i - 1} ,
\nonumber
\end{eqnarray}
whereas $\lim \limits_{ \rho \rightarrow \infty} E \left[ R_j \right]$ is
dependent on how the eigenvectors of ${\bf \Sigma}_i$ are related to
${\bf \Sigma}_j$. It is also to be noted that $E\left[ R_i \right]$
increases (and $d_{ {\bf \Sigma}_i }(  \cdot, \cdot)$ decreases) as
$\kappa_i$ increases. That is, the more ill-conditioned ${\bf \Sigma}_i$
is, the larger the high-$\snr$ statistical beamforming rate asymptote
and {\em vice versa}. This should be intuitive as our goal is only to
maximize $E \left[ R_i \right]$ and the above choice
achieves that goal.

We now consider the sum-rate setting 
restricted to the case where ${\bf \Sigma}_1$ and ${\bf \Sigma}_2$
have the same set of orthonormal eigenvectors. Instead of using the definitions of
${\bf \Sigma}_1$ and ${\bf \Sigma}_2$ as in~(\ref{pdef}), for simplicity,
we will assume that ${\bf U} = {\widetilde {\bf U}} = \left[ {\bf u}_1,
\hsppp {\bf u}_2 \right]$. We define $\kappa_1$ and $\kappa_2$ as
in~(\ref{qdef}).
Without loss in generality, we can also
assume that $\kappa_1 > 1$. Three possibilities arise depending on the
relationship between $1$, $\kappa_1$ and $\kappa_2$: i) $\kappa_1 > 1
\geq \kappa_2$, ii) $\kappa_1 > \kappa_2 > 1$, and iii) $\kappa_2 \geq
\kappa_1 > 1$. (Note that first case subsumes the setting where
$\mu_1 = \mu_2 = \mu$ and ${\bf \Sigma}_2 = \mu {\bf I}$.) The main
result is the following theorem.
\begin{thm}
\label{thm1}
The sum-rate is maximized by the following choice of beamforming vectors:
\begin{eqnarray}
\begin{array}{cc}
{\bf w}_{1, \hsppp {\sf opt}} = e^{j \nu_1} {\bf u}_1,
\hspp \hspp {\bf w}_{2, \hsppp {\sf opt}} =
e^{j \nu_2} {\bf u}_2 & {\rm if}
\hspp \hspp {\rm i)} \hspp \hspp {\rm or} \hspp \hspp {\rm ii)}
\hspp \hspp {\rm is} \hspp \hspp {\rm true}, \\
{\bf w}_{1, \hsppp {\sf opt}} = e^{j \nu_2} {\bf u}_2,
\hspp \hspp {\bf w}_{2, \hsppp {\sf opt}} =
e^{j \nu_1} {\bf u}_1 & {\rm if} \hspp \hspp {\rm iii)} \hspp
\hspp {\rm is} \hspp \hspp {\rm true}
\end{array}
\nonumber
\end{eqnarray}
for some choice of $\nu_i \in [0, 2\pi), \hspp i = 1,2$. The optimal
sum-rate is given as
\begin{eqnarray}
E\left[R_1 \right] + E \left[ R_2 \right] \stackrel{\rho \rightarrow
\infty} {\rightarrow} \left\{
\begin{array}{cc}
\frac{\kappa_1 \hsppp \cdot \hsppp \log \left(\kappa_1\right)}{\kappa_1-1} +
\frac{ \log \left(\kappa_2 \right)}{\kappa_2-1} & {\rm if} \hspp
\kappa_1 \geq \kappa_2 \\
\frac{\kappa_2 \hsppp \cdot \hsppp \log \left( \kappa_2 \right)}{\kappa_2-1} +
\frac{\log \left(\kappa_1 \right)} {\kappa_1-1} & {\rm if} \hspp
\kappa_1 < \kappa_2
\end{array}
\right.
\nonumber
\end{eqnarray}
\end{thm}
{\vspace{0.05in}}
\begin{proof}
The proof follows by decomposing ${\bf w}_1$ and ${\bf w}_2$ along
the obvious orthogonal basis of $\{ {\bf u}_1, {\bf u}_2 \}$:
\begin{eqnarray}
{\bf w}_1 =  \alpha {\bf u}_1 + \beta {\bf u}_2 , \hspp
{\bf w}_2 & = & \gamma {\bf u}_1 + \delta {\bf u}_2 \nonumber
\end{eqnarray}
for some choice of $\{ \alpha, \beta, \gamma, \delta \}$ with
$\alpha = |\alpha| e^{j \theta_{\alpha}}$ (similarly, for other
quantities) satisfying $| \alpha |^2 + |\beta |^2 = |\gamma|^2 +
| \delta |^2 = 1$. A direct optimization of the high-$\snr$ sum-rate
expression shows that $\{ \theta_{\bullet} \}$ enters the optimization
only via the term $| \beta \gamma - \alpha \delta |$, which can be
maximized by setting $\theta_{\alpha} + \theta_{\delta} - \theta_{\beta}
- \theta_{\gamma} = \pi$ (modulo $2\pi$). Parameterizing
$|\alpha|$ and $|\gamma|$ as $|\alpha| = \sin(\theta)$ and $|\gamma| =
\sin(\phi)$ for some $\{ \theta, \phi \} \in [0, \pi/2]$, we can
show that the sum-rate is maximized by $\theta = \pi/2$ and $\phi = 0$
if i) or ii) is true and by $\theta = 0$ and $\phi = \pi/2$ if iii)
is true.
For this, we establish an upper bound to the sum-rate and show that this
bound is achieved by the choice as in the statement of the theorem.
\end{proof}

We now consider the general case where ${\bf \Sigma}_1$ and ${\bf \Sigma}_2$
do not have the same set of eigenvectors.
\begin{thm}
\label{thm2}
In the general case, the sum-rate is maximized
\begin{eqnarray}
{\bf w}_{1, \hsppp {\sf opt}} = e^{j \nu_1} {\bf u}_1 \big(
{\bf \Sigma}_2^{-1} \hsppp {\bf \Sigma}_1 \big),
\hspp \hspp {\bf w}_{2, \hsppp {\sf opt}} =
e^{j \nu_2} {\bf u}_1 \big( {\bf \Sigma}_1^{-1} \hsppp
{\bf \Sigma}_2 \big)
\nonumber
\end{eqnarray}
for some choice of $\nu_i \in [0, 2\pi), \hspp i = 1,2$.
\end{thm}
\begin{proof}
For this case, we define ${\bf \Sigma}$ and its corresponding
eigen-decomposition as
\begin{eqnarray}
{\bf \Sigma} \triangleq {\bf \Sigma}_2^{-\frac{1}{2}} \hsppp
{\bf \Sigma}_1 \hsppp {\bf \Sigma}_2^{-\frac{1}{2}} =
{\bf V} \hspp {\sf diag}\left([ \eta_1 \hspp \eta_2  ]\right) \hspp
{\bf V}^H \nonumber
\end{eqnarray}
where ${\bf V} = \left[ {\bf v}_1 \hspp {\bf v}_2 \right]$ and
$\eta_1 \geq \eta_2$. Since ${\bf \Sigma}_2$ is a full rank matrix and
$\{ {\bf v}_1, {\bf v}_2 \}$ form a basis, the vectors
${\bf \Sigma}_2^{-\frac{1}{2} } \hsppp {\bf v}_1$ and
${\bf \Sigma}_2^{-\frac{1}{2} } \hsppp {\bf v}_2$ also form a basis (albeit
non-orthogonal, in general). We can decompose ${\bf w}_1$ and ${\bf w}_2$
along these vectors as
\begin{eqnarray}
{\bf w}_1 = \frac{
\alpha {\bf \Sigma}_2^{-\frac{1}{2} } {\bf v}_1 +
\beta{\bf \Sigma}_2^{-\frac{1}{2} } {\bf v}_2 }
{ \| \alpha{\bf \Sigma}_2^{-\frac{1}{2} }  {\bf v}_1 +
\beta {\bf \Sigma}_2^{-\frac{1}{2} }  {\bf v}_2 \| },
\hsppp {\bf w}_2 =
\frac{ \gamma {\bf \Sigma}_2^{-\frac{1}{2} } {\bf v}_1 +
\delta {\bf \Sigma}_2^{-\frac{1}{2} }  {\bf v}_2}
{ \| \gamma {\bf \Sigma}_2^{-\frac{1}{2} } {\bf v}_1 +
\delta {\bf \Sigma}_2^{-\frac{1}{2} }  {\bf v}_2 \| }
\nonumber
\end{eqnarray}
for some choice of $\{ \alpha, \beta, \gamma, \delta \}$ with
$\alpha = |\alpha| e^{j \theta_{\alpha}}$ (similarly, for other
quantities) satisfying $| \alpha |^2 + |\beta |^2 = |\gamma|^2 +
| \delta |^2 = 1$. A suitable coordinate transformation at this stage
results in an optimization problem that is related to the special case
of Theorem~\ref{thm1}. After this transformation, the proof follows
along the same logic as in Theorem~\ref{thm1}.
\end{proof}
The reason for the peculiar choice of decomposition in the above proof
(instead of decomposing the beamforming vectors along $\{ {\bf v}_1,
{\bf v}_2 \}$) is that ${\bf \Sigma}_2^{-\frac{1}{2} } \hsppp
{\bf v}_i, \hsppp i = 1,2$ turn out to be the dominant generalized
eigenvectors of the pairs $\left( {\bf \Sigma}_1, \hsppp {\bf \Sigma}_2
\right)$ and $\left( {\bf \Sigma}_2, \hsppp {\bf \Sigma}_1 \right)$,
respectively. Recall from Footnote~\ref{fn_general} the definition of a
generalized eigenvector. For the above claim, note that
\begin{eqnarray}
{\bf \Sigma}_2^{-1} \hsppp {\bf \Sigma}_1 & = &  {\bf \Sigma}^{ - \frac{1}{2} }
\hsppp \Big( {\bf V} \hspp {\sf diag}\left([ \eta_1 \hspp \eta_2  ]\right)
\hspp {\bf V}^H \Big) \hsppp {\bf \Sigma}_2^{\frac{1}{2}}
= {\bf M} \hsppp {\bf D} \hsppp {\bf M}^{-1} \nonumber \\
{\bf \Sigma}_1^{-1} \hsppp {\bf \Sigma}_2 & = &
\Big({\bf \Sigma}_2^{-1} \hsppp {\bf \Sigma}_1  \Big)^{-1} =
{\bf M} \hsppp {\bf D}^{-1} \hsppp {\bf M}^{-1} \nonumber
\end{eqnarray}
where ${\bf M} = {\bf \Sigma}_2^{-\frac{1}{2} } \hsppp {\bf V}$ and
${\bf D} = {\sf diag}\left([ \eta_1 \hspp \eta_2  ]\right)$.
Theorem~\ref{thm1} is indeed a special case of Theorem~\ref{thm2}. For
this, note that the dominant eigenvector of ${\bf \Sigma}_2^{-1}
\hsppp {\bf \Sigma}_1$ is ${\bf u}_1$ and ${\bf u}_2$ when $\kappa_1 >
\kappa_2$ and $\kappa_2 < \kappa_1$, respectively.

\section{Ergodic Sum-Rate: General $M$ Case}
A recent advance~\cite{hammarwall1,hammarwall2} allows a computation
of the density function of weighted sum of standard central chi-squared
terms ({\em generalized} chi-squared random variables). Alternate to the
approach of Prop.~\ref{prop1},
this approach allows closed-form expressions in the general $M$
case. For example, if ${\bf \Lambda}_i(j), \hsppp j = 1,
\cdots, M$ are distinct\footnote{More complicated expressions can be
obtained in case $\{ {\bf \Lambda}_i(j) \}$ are not distinct.}, we have
\begin{align}
& E \left[ I_{i, \hsppp 1} \right] =
\sum_{k=1}^M \prod_{j = 1, \hsppp j \neq k}^M \frac{ {\bf \Lambda}_i(k)}
{  {\bf \Lambda}_i(k) - {\bf\Lambda}_i(j)} \cdot x_k
\\
& {\hspace{0.1in}}
x_k = \exp \left( \frac{\rho}{ {\bf \Lambda}_i(k) M} \right)
E_1 \left( \frac{\rho}{ {\bf \Lambda}_i(k) M} \right).
\nonumber
\end{align}
For $E \left[  I_{i, \hsppp 2} \right]$, replace ${\bf \Lambda}_i$ by
${\widetilde{ \bf \Lambda}}_i$. It can be checked that this expression
matches with the expression in the $M = 2$ case.

Nevertheless, it is important to note that the formula
above is in terms of the eigenvalue matrices
$\{ {\bf \Lambda}_i, {\widetilde{ \bf \Lambda}}_i, \hsppp i = 1,
\cdots, M \}$, which become harder (and impossible for $M \geq 5$)
to compute in closed-form as a function of the beamforming vectors
and the covariance matrices as $M$ increases. Approximation
to the generalized chi-squared random variable by a Gamma distribution
with matching first two moments 
can also be used to produce sum-rate approximations. However, these
approximations are of similar complexity as the above formula.
In contrast, we now provide asymptotic approximations to the sum-rate
directly in terms of the relevant variables.
\begin{prp}
\label{prop_asy}
For any fixed $\rho$, the ergodic information-theoretic rate
achievable at user $i$ (where $i = 1, \cdots, M$) converges as
$M \rightarrow \infty$ to
\begin{eqnarray}
E \left[ R_i \right] & \rightarrow & \log \left( 1 + {\sf SINR}_i
\right) \triangleq
{\cal R}_{i, \hsppp \infty}
\nonumber \\
{\sf SINR}_i & = &
\frac{ \frac{\rho}{M} \cdot {\bf w}_i^H {\bf \Sigma}_i {\bf w}_i }
{1  + \frac{\rho}{M} \cdot \sum_{j = 1, \hsppp j \neq i} ^M
{\bf w}_j^H {\bf \Sigma}_i {\bf w}_j  } \triangleq
\frac{ {\sf S}_i } { {\sf I}_i}. \nonumber
\end{eqnarray}
\end{prp}
{\vspace{0.1in}}
\begin{proof}
The proof follows along a law of large numbers-type argument,
strengthened to convergence in mean via a suitable truncation
technique. 
\end{proof}
\begin{prp}
\label{prop_opt_asy}
Based on the above expression, we have the following conclusions
that mirror the main results of Sec.~\ref{sec3}. i) We have the following
bound for $\sum_{i=1}^M {\cal R}_{i, \hsppp \infty}$:
\begin{eqnarray}
1 - \frac{\rho}{M} \cdot \max_{i = 1, \cdots, M} \sum_{j=1}^M
{\bf w}_j^H {\bf \Sigma}_i {\bf w}_j
\leq \frac{ \sum_{i=1}^M {\cal R}_{i, \hsppp \infty} }
{ \frac{\rho}{M} \cdot \sum_{i=1}^M {\bf w}_i^H {\bf \Sigma}_i
{\bf w}_i } \leq 1. \nonumber
\end{eqnarray}
Thus, the optimal beamforming vectors as $\rho \rightarrow 0$ are
such that ${\bf w}_{i, \hsppp {\sf opt} } = {\bf u}_1( {\bf \Sigma}_i ),
\hsppp i = 1, \cdots, M$. ii) For any $\rho$, we have
\begin{eqnarray}
{\cal R}_{i, \hsppp \infty} \leq \log \left( 1 +
\frac{ \frac{\rho}{M} \cdot \lambda_1({\bf \Sigma}_i) }
{1 + \frac{\rho}{M} \cdot \sum_{j=2}^M \lambda_j({\bf \Sigma}_i)}
\right) \nonumber
\end{eqnarray}
and ${\cal R}_{i,\hsppp \infty}$ is maximized by
${\bf w}_{i, \hsppp {\sf opt}} = {\bf u}_1( {\bf \Sigma}_i)$, and
\begin{eqnarray}
\Big\{ {\bf w}_{j, \hsppp {\sf opt}}, \hsppp j = 1, \cdots, M, j \neq i
\Big\}
= \Big\{ {\bf u}_{j}({\bf \Sigma}_i), \hsppp j = 2, \cdots, M \Big\}.
\label{poss} \nonumber
\end{eqnarray}
{\vspace{0.05in}}
iii) $\sum_{i=1}^M {\cal R}_{i, \hsppp \infty}$ is optimized by the
set of beamforming vectors that solve the following fixed-point equations:
\begin{eqnarray}
\frac{ {\bf \Sigma}_i {\bf w}_i }{ {\sf I}_i \cdot (1 + {\sf SINR}_i )}
- \sum_{j \neq i} \frac{ {\sf SINR}_j \cdot {\bf \Sigma}_j {\bf w}_i }
{ {\sf I}_j \cdot (1 + {\sf SINR}_j) } = {\bf 0} , \hspp i = 1, \cdots, M.
\nonumber
\end{eqnarray}
\endproof
\end{prp}

\section{Conclusion}
We have studied statistics-based linear beamformer design for the MISO
broadcast channel in this work. Based on a closed-form computation of the
ergodic sum-rate in the $M = 2$ (two-user) case, we provide intuition on
the structure of the optimal beamforming vectors that maximize the sum-rate
in the low- and the high-$\snr$ extremes. While further intuition on the
small $M$ case seems difficult, in the asymptotics of $M$, we are able to
obtain intuition on the structure of the optimal beamforming vectors. The
case of optimal statistical linear beamforming design has not received
much attention in the literature and our work sets the course for
a systematic and low-complexity limited feedback design in the broadcast
setting, which is of considerable importance in the standardization efforts.

\section*{Acknowledgment}
This work has been supported in part by the NSF through grant CNS-0831670
at the University of Illinois.


\begin{thebibliography}{10}

\bibitem{caire_shamai}
G.~Caire and S.~Shamai,
\newblock ``{On the achievable throughput of a multiantenna Gaussian broadcast
  channel},''
\newblock {\em IEEE Trans. Inf. Theory}, vol. 49, no. 7, pp. 1691--1706, July
  2003.

\bibitem{pramodv}
P.~Viswanath and D.~N.~C. Tse,
\newblock ``{Sum capacity of the vector Gaussian broadcast channel and
  downlink-uplink duality},''
\newblock {\em IEEE Trans. Inf. Theory}, vol. 49, no. 8, pp. 1912--1921, Aug.
  2003.

\bibitem{jindal1}
S.~Vishwanath, N.~Jindal, and A.~Goldsmith,
\newblock ``{Duality, achievable rates and sum rate capacity of Gaussian MIMO
  broadcast channel},''
\newblock {\em IEEE Trans. Inf. Theory}, vol. 49, no. 10, pp. 2658--2668, Oct.
  2003.

\bibitem{jindal2}
N.~Jindal, S.~Vishwanath, and A.~Goldsmith,
\newblock ``{On the duality of Gaussian multiple-access and broadcast
  channels},''
\newblock {\em IEEE Trans. Inf. Theory}, vol. 50, no. 5, pp. 768--783, May
  2004.

\bibitem{wei_yu}
W.~Yu and J.~M. Cioffi,
\newblock ``{Sum capacity of Gaussian vector broadcast channels},''
\newblock {\em IEEE Trans. Inf. Theory}, vol. 50, no. 9, pp. 1875--1892, Sept.
  2004.

\bibitem{weingarten}
H.~Weingarten, Y.~Steinberg, and S.~Shamai,
\newblock ``{The capacity region of the Gaussian multiple-input multiple-output
  broadcast channel},''
\newblock {\em IEEE Trans. Inf. Theory}, vol. 52, no. 9, pp. 3936--3964, Sept.
  2006.

\bibitem{boche1}
M.~Schubert and H.~Boche,
\newblock ``{Solution of multiuser downlink beamforming problem with individual
  SINR constraint},''
\newblock {\em IEEE Trans. Veh. Tech.}, vol. 53, no. 1, pp. 18--28, Jan. 2004.

\bibitem{swindlehurst}
C.~Peel, B.~Hochwald, and A.~Swindlehurst,
\newblock ``{Vector perturbation techniques for near-capacity multiantenna
  multi-user communication},''
\newblock {\em IEEE Trans. Commun.}, vol. 53, no. 1, pp. 195--202, Jan. 2005.

\bibitem{wiesel}
A.~Wiesel, Y.~C. Eldar, and S.~Shamai,
\newblock ``{Zero forcing precoding and generalized inverses},''
\newblock {\em IEEE Trans. Sig. Proc.}, vol. 56, no. 9, pp. 4409--4418, Sept.
  2008.

\bibitem{tareq}
T.~Y. Al-Naffouri, M.~Sharif, and B.~Hassibi,
\newblock ``{How much does transmit correlation affect the sum-rate scaling of
  MIMO Gaussian broadcast channels?},''
\newblock {\em IEEE Trans. Commun.}, vol. 57, no. 2, pp. 562--572, Feb. 2009.

\bibitem{trivellato}
M.~Trivellato, F.~Boccardi, and H.~Huang,
\newblock ``{On transceiver design and channel quantization for downlink
  multiuser MIMO systems with limited feedback},''
\newblock {\em IEEE Journ. Sel. Areas in Commun.}, vol. 6, no. 8, pp.
  1494--1504, Oct. 2008.

\bibitem{vasanth_isit_rhv}
V.~Raghavan, M.~L. Honig, and V.~V. Veeravalli,
\newblock ``{Performance analysis of RVQ codebooks for limited feedback
  beamforming},''
\newblock {\em Proc. IEEE Intern. Symp. Inf. Theory}, pp. 2437--2441, July
  2009.

\bibitem{coord_bf}
C.-B. Chae, D.~Mazzarese, N.~Jindal, and R.~W. Heath, Jr.,
\newblock ``{Coordinated beamforming with limited feedback in the MIMO
  broadcast channel},''
\newblock {\em IEEE Journ. Sel. Areas in Commun.}, vol. 26, no. 8, pp.
  1505--1515, Oct. 2008.

\bibitem{gradshteyn}
I.~S. Gradshteyn and I.~M. Ryzhik,
\newblock {\em {Table of Integrals, Series, and Products}},
\newblock Academic Press, NY, 4th edition, 1965.

\bibitem{hammarwall1}
D.~Hammarwall, M.~Bengtsson, and B.~E. Ottersten,
\newblock ``{Acquiring partial CSI for spatially selective transmission by
  instantaneous channel norm feedback},''
\newblock {\em IEEE Trans. Sig. Proc.}, vol. 56, no. 3, pp. 1188--1204, Mar.
  2008.

\bibitem{hammarwall2}
D.~Hammarwall, M.~Bengtsson, and B.~E. Ottersten,
\newblock ``{Utilizing the spatial information provided by channel norm
  feedback in SDMA systems},''
\newblock {\em IEEE Trans. Sig. Proc.}, vol. 56, no. 7-2, pp. 3278--3293, July
  2008.

\end{thebibliography}

\end{document}